\documentclass{article}

\usepackage[]{datetime}

\usepackage{cite}

\usepackage{amssymb}
\usepackage{graphicx}
\usepackage[utf8x]{inputenc}
\usepackage{enumerate}
\usepackage{amsfonts}
\usepackage{amsmath}
\usepackage{stmaryrd}
\usepackage{xspace}
\usepackage{tikz}
\usepackage{tikz}
\tikzstyle{every node} = [draw, circle, fill = black, minimum size =
4pt, inner sep = 0pt]
\tikzstyle{normal} = [draw=none, fill = none]

\usepackage{hyperref}
\usepackage{amsthm}

\newcommand{\ie}{\emph{i.e.}}
\newcommand{\etc}{\emph{etc.}}
\newcommand{\eg}{\emph{e.g.}}
\newcommand{\etal}{\emph{et al.}}

\newcommand{\floor}[1]{\left\lfloor~#1~\right\rfloor}
\newcommand{\ceil}[1]{\left\lceil~#1~\right\rceil}
\newcommand{\third}[1]{\frac{#1}{3}} \newcommand{\paren}[1]{\left( #1
\right)} \newcommand{\acc}[1]{\left\{ #1 \right\}}

\newcommand{\N}{\mathbb{N}} 
\newcommand{\intv}[2]{\llbracket #1, #2 \rrbracket}
\newcommand{\tw}{\mathbf{tw}}
\newcommand{\pw}{\mathbf{pw}}
\newcommand{\dg}{\mathrm{deg}} 
\newcommand{\multidg}{\mathrm{{\rm deg}^m}} 
\newcommand{\vertices}[1]{{V(#1)}}
\newcommand{\edges}[1]{{E(#1)}}
\newcommand{\induced}[2]{#1[#2]} 
\newcommand{\mindeg}{\delta} \newcommand{\dege}{\delta^*} 
\newcommand{\degec}{\delta_c} 
\newcommand{\card}[1]{\left | #1 \right |} 
\newcommand{\subgraph}{\subseteq} \newcommand{\nb}{\mathrm{N}}
\newcommand{\lminor}{\leq_\mathrm{m}} 
\newcommand{\contains}{\geq_\mathrm{m}} 
\newcommand{\notcontain}{\not \geq_\mathrm{m}} 
\newcommand{\K}[1]{\mathrm{K}_{#1}} 
\newcommand{\sg}[1]{\mathrm{\Xi}_{#1}} 

\newcommand{\half}{\frac{1}{2}} \newcommand{\ad}{\mathrm{ad}}
\newcommand{\mult}{\mathrm{mult}} 
\renewcommand{\leq}{\leqslant}
\renewcommand{\geq}{\geqslant}

\newcommand{\pack}[1]{{\mathbf{pack}_{#1}}}
\newcommand{\cover}[1]{{\mathbf{cover}_{#1}}}
\newcommand{\desc}[1]{\mathrm{desc}_{#1}}

\newtheorem{theorem}{Theorem}
\newtheorem{prop}{Proposition}
\newtheorem{lemma}{Lemma}

\theoremstyle{remark} \newtheorem{remark}{Remark}
\theoremstyle{definition} \newtheorem{definition}{Definition}

\def\claimb{$$\vcenter\bgroup\advance\hsize by -8em\noindent
\refstepcounter{claimb}\ignorespaces\it} \makeatletter
\def\endclaimb{\rm\egroup\leqno(\theclaimb)$$\global\@ignoretrue}
\makeatother

\newenvironment{proofclaim}[1][]{
\noindent \emph{Proof~#1.}{}{}{}} {\hfill$\Diamond$\vspace{1em}}

\begin{document}

\title{Polynomial Gap Extensions of the Erdős--Pósa Theorem}

\author{Jean-Florent Raymond\thanks{Emails:
\href{mailto:jeanflorent.raymond@ens-lyon.fr}{{\sf
jeanflorent.raymond@ens-lyon.fr}},  
\href{mailto:sedthilk@thilikos.info}{{\sf
sedthilk@thilikos.info}}}~\thanks{LIRMM, Montpellier,
France.}\and Dimitrios
M. Thilikos$^*$\thanks{Department of Mathematics,
National and Kapodistrian University of Athens and CNRS (LIRMM). Co-financed by the E.U. (European Social Fund - ESF) and Greek national funds through the Operational Program ``Education and Lifelong Learning'' of the National Strategic Reference Framework (NSRF) - Research Funding Program: ``Thales. Investing in knowledge society through the European Social Fund''.}}
%
\date{\empty}

\maketitle

\abstract{\noindent Given a graph $H$,
we denote by ${\cal M}(H)$ all graphs that can be contracted to
$H$. The following extension of the Erdős--Pósa theorem holds: for
every $h$-vertex planar graph $H$, there exists a function $f_{H}$ such
that every graph $G$, either contains $k$ disjoint copies of graphs in
${\cal M}(H)$, or contains a set of $f_{H}(k)$ vertices meeting
every subgraph of $G$ that belongs in ${\cal M}(H)$. In this paper we
prove that this is the case for
every graph $H$ of pathwidth at most $2$ and, in particular, that
$f_{H}(k) = 2^{O(h^2)}\cdot k^{2}\cdot \log k$. As a main ingredient
of the proof of our result, we show that for every graph $H$ on $h$
vertices and pathwidth at most $2$, either $G$ contains $k$ disjoint
copies of $H$ as a minor or the treewidth of $G$ is upper-bounded by
$2^{O(h^2)}\cdot k^{2}\cdot \log k$. We finally prove that the
exponential dependence on $h$ in these bounds can be avoided if
$H=K_{2,r}$. In particular, we show that $f_{K_{2,r}}=O(r^2\cdot
k^2)$.}

\paragraph{Keywords:} Treewidth, Graph Minors, Erdős--Pósa Theorem

\section{Introduction}

In 1965, Paul Erdős and Lajos Pósa proved that every graph that does
not contain $k$ disjoint cycles, contains a set of $O(k\log k)$
vertices meeting all its cycles~\cite{ErdosP65}. Moreover, they gave a
construction asserting that this bound is tight. This classic result
can be seen as a ``loose'' min-max relation between covering and
packing of combinatorial objects. Various extensions of this result,
referring to different notions of packing and covering, attracted the
attention of many researchers in modern Graph Theory (see,
\eg~\cite{ErdosP07thee,GeelenK09thee}).

Given a graph $H$, we denote by ${\cal M}(H)$ the set of all graphs
that can be contracted to $H$ (\ie~if $H'\in {\cal M}(H)$, then $H$
can be obtained from $H'$ after contracting edges). We call the
members of ${\cal M}(H)$ {\em models} of $H$. Then the notions of
covering and packing can be extended as follows: we denote by ${\bf
cover}_{H}(G)$ the minimum number of vertices that meet every model of
$H$ in $G$ and by ${\bf pack}_{H}(G)$ the maximum number of
mutually disjoint models of $H$ in $G$. We say that a graph $H$ {\em
has the Erdős--Pósa Property} if there exists a function $f_{H} \colon \N\rightarrow \N$ such
that for every graph $G$,
\begin{eqnarray} \text{if}\ k={\bf pack}_{H}(G),\ \text{then}\ k \leq
{\bf cover}_{H}(G) \leq f_{H}(k)\label{tlp4o}
\end{eqnarray}

We will refer to $f_{H}$ as the {\em gap} of the {Erdős--Pósa}
Property. Clearly, if $H=K_{3}$, then~\eqref{tlp4o} holds for
$f_{K_{3}}=O(k \log k)$ and the general question is to find, for
each instantiation of $H$, the best possible estimation of the gap
$f_{H}$, if it exists.

It turns out that $H$ has the Erdős--Pósa Property if and only if $H$ is
a planar graph. This beautiful result appeared as a byproduct of the
Graph Minors series of Robertson and Seymour. In particular, it is a
consequence of the grid-exclusion theorem, proved
in~\cite{RobertsonS86GMV} (see also~\cite{DIE10B}).
\begin{prop}
\label{classic} There is a function $g \colon \N\rightarrow \N$
such that if a graph excludes an $r$-vertex planar graph $R$ as a
minor, then its treewidth is bounded by $g(r)$. \end{prop}
\noindent In~\cite{RobertsonS86GMV} Robertson, Seymour, and Thomas
conjectured that $g$ is a low degree polynomial function. Currently,
the best known bound for $g$ is $g(k)=2^{O(k\log k)}$ and follows
from~\cite{DiestelJGT99high} and~\cite{LeafS12sube} (see
also~\cite{RobertsonS86GMV,KawarabayashiK12line} for previous proofs
and improvements). As the function $g$ is strongly used in the
construction of the function $f_{H}$ in~\eqref{tlp4o}, the best, so
far, estimation for $f_{H}$ is far from being exponential in general.
This initiated a quest for detecting instantiations of $H$ where a
polynomial gap $f_{H}$ can be proved.

The first result in the direction of proving polynomial gaps for the
{Erdős--Pósa} Property appeared in~\cite{FominLMPS13quad} where $H$ is
the graph $\theta_{c}$ consisting of two vertices connected by $c$
multiple edges (also called {\em $c$-pumpkin graph}). In particular,
in~\cite{FominLMPS13quad} it was proved that
$f_{\theta_{c}}(k)=O(c^2k^{2})$. More recently Fiorini, Joret, and Sau
optimally improved this bound by proving that $f_{\theta_{c}}(k)\leq
c_{t}\cdot k\cdot \log k$ for some computable constant $c_{t}$
depending on $c$~\cite{FioriniJS13}. In~\cite{FioriniJW12excl}
Fiorini, Joret, and Wood proved that if $T$ is a tree, then
$f_{T}(k)\leq c_{T}\cdot k$ where $c_{T}$ is some computable constant
depending on $T$. Finally, very recently, Fiorini ~\cite{FioriniHJ13}
proved that $f_{K_{4}}=O(k\log k)$.

Our main result is a polynomial bound on $f_{H}$ for a broad family of
planar graphs, namely those of pathwidth at most $2$. We prove the
following:
\begin{theorem}
\label{first} If $H$ is an $h$-vertex graph of pathwidth at most $2$ and $h > 5$,
then~\eqref{tlp4o} holds for $f_{H}(k)=2^{O(h^2)}\cdot k^2 \cdot \log k$.
\end{theorem} Note that the contribution of $h$ in $f_{H}$ is
exponential. However, such a dependence can be waived when we restrict
$H$ to be $K_{2,r}$. Our second result is the following:

\begin{theorem}
\label{second} If $H=K_{2,r}$, then~\eqref{tlp4o} holds for $f_{H}(k)=
O(r^2 \cdot k^2$).
\end{theorem}

Both results above are based on a proof of Proposition~\ref{classic},
with polynomial $g$, for the cases where $R$ consists of $k$ disjoint
copies of $H$ and $H$ is either a graph of pathwidth at most $2$ or
$H=K_{2,3}$ (Theorems~\ref{th1} and~\ref{th2} respectively). For
this, we follow an approach that makes strong use of the $k$-mesh
structure introduced by Diestel \etal~\cite{DiestelJGT99high} in their
proof of Proposition~\ref{classic}. Our proof indicates that, when
excluding copies of some graph of pathwidth at most $2$, the entangled
machinery of~\cite{DiestelJGT99high} can be partially modified so that
polynomial bounds on treewidth are possible. Finally, these bounds are
then ``translated'' to polynomial bounds for the Erdős--Pósa gap using
a technique developed in~\cite{FominST11stre} (see
also~\cite{FominLMPS13quad}).

\section{Definitions and notations}
\label{defs}

\subsection{Basics}

In this paper, logarithms are binary.

\paragraph{Graphs and subgraphs}

A \emph{graph}~$G$ is a pair~$(V,E)$ where~$V$ is called the set of
vertices of~$G$ and~$E$ is called the set of edges of~$G$ and
satisfies~$E \subseteq V^2$. Two vertices~$v,u$ of~$G$ are said to be
\emph{adjacent} if~$(u,v) \in E$. A \emph{multigraph} is a graph where
multiple edges between two vertices are allowed.
In this paper, the graphs we consider are finite, undirected and
without loops. Unless otherwise specified, graphs are assumed to be
simple (\ie~multiedges are not allowed).

For any graph~$G$, $\vertices{G}$ (resp.\ $\edges{G}$) denotes the set
of vertices (resp.\ edges) of~$G$.
A graph~$G'$ is a \emph{subgraph} of a graph~$G$ if $\vertices{G'}
\subseteq \vertices{G}$ and $\edges{G'} \subseteq \edges{G}$ and we
write it~$G' \subgraph G$.
If~$X$ is a subset of~$\vertices{G}$, we note~$\induced{G}{X}$ the
\emph{subgraph of~$G$ induced by~$X$}, \ie~the graph $(X, \{xy \in
\edges{G},\ x \in X\ \mathrm{and}\ y \in X\})$.

When talking about graphs, unless otherwise stated, by \emph{disjoint}
we mean vertex-disjoint.
We denote by \emph{$\K{n}$} the complete graph on~$n$ vertices and by~$\K{p,q}$ the complete bipartite graph with partitions of size~$p$ and~$q$. For any integer~$k$ and any graph~$G$, the graph~$k \cdot G$ is the disjoint union of~$k$ copies of the graph~$G$.
A pair $\acc{A,B}$ is a separation of a graph $G$ if $A \cup B =
\vertices(G)$ and $G$ has no edge between $A \setminus B$ and $B
\setminus A.$ The integer $\card{A \cap B}$ is the \emph{order} of the
separation $\acc{A,B}.$
We assume that the reader is familiar with the basic graph classes:
paths, cycles, trees,~\etc.

\paragraph{Neighbourhood and degree}
For any vertex~$v \in \vertices{G}$, the \emph{neighbourhood}~$\nb_G(v)$ of~$v$ in~$G$ is the set of vertices that are adjacent to~$v$ in~$G$. The \emph{degree} of~$v\in \vertices{G}$ in~$G$, denoted~$\dg_G(v)$, is the cardinal of~$\nb_G(v)$. The minimum value taken by~$\dg_G$ in~$\vertices{G}$ is called the \emph{minimum degree} of~$G$ and denoted by~$\mindeg(G)$. When dealing with multigraphs, the \emph{multidegree} of a vertex~$v$ (written $\multidg(v)$) is the number of simple edges incident to~$v$. In these notations, we drop the subscript when it is obvious. The average degree over all vertices of a graph~$G$ is written~$\ad(G)$.

\paragraph{Contractions}
In a graph $G$, a \emph{contraction} of the edge~$e = (u,v) \in
\edges{G}$ is the operation that transforms $G$ into a graph~$H$ such
that $\vertices{H} = \vertices{G} \backslash \{u,v\} \cup \{v_e\}$ and
$\edges{H} = \{(x,y) \in \edges{G},\ x \not \in \{u,v\}\ \mathrm{and}\
y \not \in \{u,v\}\} \cup \{(x,v_e),\ (x,u) \in \edges{G}\
\mathrm{or}\ (x,v) \in \edges{G}\}$.
We say that a graph~$G$ can be \emph{contracted} to a graph~$H$ if~$H$
is the result of a sequence of edge contractions on~$G$.

\paragraph{Trees}
An acyclic connected graph is called a \emph{tree}. The vertices of
degree 1 of a tree are its \emph{leaves} and its other vertices are
called \emph{internal vertices}. A tree whose every internal vertex
has degree at most 3 is said to be \emph{ternary}. A \emph{binary
  tree} is a ternary tree whose one of the internal nodes, the
\emph{root}, is distinguished and has degree at most 2.

\subsection{More definitions}

\begin{definition}[graph $\sg{r}$]
 We define the graph $\Xi_{r}$ as the graph of the following form (see figure~\ref{fig:sg}).
\[
\left \{
  \begin{array}{l}
    \vertices{G} = \{x_0, \dots, x_{r-1},y_0, \dots, y_{r-1},z_0,
    \dots, z_{r-1}\}\\
    \edges{G} = \{(x_i,x_{i+1}), (z_i,z_{i+1})\}_{i \in \intv{0}{r-2}}
    \cup \{(x_i,y_{i}),(y_i,z_i)\}_{i \in \intv{0}{r-1}}
  \end{array}
\right .
\]  
\end{definition}

\begin{figure}[ht]
  \centering
  \begin{tikzpicture}
    \foreach \x in {0,...,4}
    {
      \draw (\x,0) node {};
      \node[normal] at (\x + .25, .25) {$z_{\x}$};
    }
    \foreach \x in {0,...,4}
    {
      \draw (\x,1) node {};
      \node[normal] at (\x + .25, 1.25) {$y_{\x}$};
    }
    \foreach \x in {0,...,4}
    {
      \draw (\x,2) node {};
      \node[normal] at (\x + .25, 2.25) {$x_{\x}$};
    }
    \foreach \x in {0,...,4}
    {
      \draw (\x,0) -- (\x,1);
      \draw (\x,1) -- (\x,2);
      \draw (\x,0) -- (+1,0);
      \draw (\x,2) -- (+1,2);
    }
  \end{tikzpicture}
\caption{The graph $\sg{5}$}
  \label{fig:sg}
\end{figure}

\begin{definition}[minor model]
  A \emph{minor model} (sometimes abbreviated \emph{model}) of a graph
  $H$ in a graph $G$ is a pair $({\cal M},\varphi)$ 
  where ${\cal M}$ is a collection of disjoint subsets
of $\vertices{G}$ such that $\forall X \in {\cal M}$, $\induced{G}{X}$
is connected and $\varphi \colon \vertices{H} \to {\cal M}$ is a
bijection that satisfies $\forall \{u,v\} \in \edges{H}, \exists u'
\in \varphi(u), \exists v' \in \varphi(v),\ \{u',v'\} \in
\edges{G}$. We say that a graph $H$ is a \emph{minor} of a graph $G$
($H \lminor G$) if there is a minor model of $H$ in $G$. Notice that
$H$ is a minor of $G$ if $H$ can be obtained by a subgraph of~$G$
after contracting edges.
\end{definition}

\begin{definition}[degeneracies]
  The \emph{degeneracy} of~$G$, written~$\dege(G)$, is the maximum value
  taken by~$\mindeg(G')$ over all subgraphs~$G'$ of~$G$:
  \[
  \dege(G) = \max_{G' \subseteq G} \mindeg(G')
  \]

  Similarly, the \emph{contraction degeneracy} of~$G$, introduced
  in~\cite{Bod04} and denoted~$\degec(G)$, is the maximum value
  of~$\mindeg(G')$ for all minors~$G'$ of~$G$:
  \[
  \degec(G) = \max_{G' \lminor G} \mindeg(G')
  \]
  Remark that, as a subgraph is a minor, for all graph~$G$ we
  have the following inequality
  \[\degec(G) \geq \dege(G)\]
  These definitions remains the same on multigraphs (we do not take into account the potential multiplicities of the edges).
\end{definition}

\begin{definition}[tree decomposition and treewidth]
A \emph{tree decomposition} of a graph~$G$
is a pair~$(T,\mathcal{X})$ where $T$ is a tree and
$\mathcal{X}$ a family $(X_t)_{t \in \vertices{{T}}}$ of
subsets of $\vertices{G}$ (called \emph{bags}) indexed by
elements of $\vertices{T}$ and such that
 \begin{enumerate}[(i)]
 \item $\bigcup_{t \in \vertices{{T}}} X_t = \vertices{G}$;
 \item for every edge~$e$ of~$G$ there is an element of~$\mathcal{X}$
containing both ends of~$e$;
 \item for every~$v \in \vertices{G}$, the subgraph of~${T}$
induced by $\{t \in \vertices{{T}}\mid {v \in X_t}\}$ is connected.
 \end{enumerate}

The \emph{width} of a tree decomposition~${T}$ is defined as
equal to $\max_{t \in \vertices{{T}}}~{\card{X_t} - 1}$. The
\emph{treewidth} of~$G$, written~$\tw(G)$, is the minimum width of any
of its tree decompositions.
\end{definition}

\begin{definition}[nice tree decomposition]
  A tree decomposition $(T,\mathcal{V}$ of a graph $G$ is said to be a \emph{nice} tree
  decomposition if
  \begin{enumerate}
  \item every vertex of $T$ has degree at most 3;
  \item $T$ is rooted on one of its vertices $r$ whose bag is empty ($V_r = \emptyset$);
  \item every vertex $t$ of $T$ is
    \begin{itemize}
    \item either a \emph{base node}, \ie~a leaf of $T$ whose
      bag is empty ($V_t = \emptyset$) and different from the root;

    \item or an \emph{introduce node}, \ie~a vertex with only one
      child $t'$ such that $V_{t'} = V_t \cup \{u\}$ for some $u \in
      \vertices{G}$;

    \item or a \emph{forget node}, \ie~a vertex with only one
      child $t'$ such that $V_t = V_{t'} \cup \{u\}$ for some $u \in
      \vertices{G}$;

    \item or a \emph{join node}, \ie~a vertex with two child $t_1$ and
      $t_2$ such that $V_t = V_{t_1} = V_{t_2}$.
    \end{itemize}
  \end{enumerate}
It is known that every graph has an optimal tree decomposition which
is nice \cite{Kloks94}.
\end{definition}

\begin{definition}[path decomposition and pathwidth]
  A \emph{path decomposition} of a graph~$G$ is a tree decomposition~$T$ of~$G$ such that~$T$ is a path. Its width is the width of the tree decomposition~$T$ and the \emph{pathwidth} of~$G$, written~$\pw(G)$, is the minimum width of any of its path decompositions.
\end{definition}

\begin{definition}[linked and externally~$k$-connected]
Let~$k$ be a positive integer, $G$ be a graph and $X,Y$ be two subsets
of $\vertices{G}$.

$X$ and $Y$ are said to be \emph{linked} by a path if there is a path
in~$G$ from an element of~$X$ to an element of~$Y$.

$X$ and~$Y$ are said to be \emph{$k$-connected in~$G$} if for all disjoint subsets $X' \subseteq X$ and $Y' \subseteq Y$ such that $\card{X'} = \card{Y'} \leq k$ there are~$\card{X'}$ disjoint paths between~$X'$ and~$Y'$ in~$G$.
If these paths have no internal vertices nor edges in $G[X \cup Y]$, then~$X$ and~$Y$ are said to be \emph{externally~$k$-connected in~$G$}.
If~$X = Y$,~$X$ is said to be \emph{(externally)~$k$-connected in~$G$}.
\end{definition}

\begin{definition}[$k$-mesh,~\cite{DIE10B}]
  An (ordered) pair~$(A,B)$ of subsets of~$\vertices{G}$ is a called a \emph{$k$-mesh of order s in~$G$} if $\vertices{G} = A \cup B$ and $G[A]$ contains a ternary tree~$T$ such that
  \begin{enumerate}[(i)]
  \item $A \cap B \subseteq \vertices{T}$ and $A\cap B\cap \vertices{T}$ are nodes of degree at most 2 in~$T$;
  \item at least one leaf of~$T$ is in~$A \cap B$;
  \item $\card{A \cap B} = s$;
  \item $A \cap B$ is externally~$k$-connected in~$B$.
  \end{enumerate}
\end{definition}

\section{Preliminaries}
\label{sec:pre}

\begin{prop}[\cite{DIE10B}, (12.14.5)] 
\label{p:mesh}
  Let~$G$ be a graph and let $p \geq q \geq 1$ be integers. If~$G$
  contains no~$q$-mesh of order~$p$ then~$G$ has treewidth
  less than~$p + q - 1$.
\end{prop}

\begin{prop}[follows from~\cite{DIE10B}, (2.14.6)] 
  \label{p:tree_cut}
  Let~$k \geq 2$ be an integer. Let~$T$ be a tree of maximum degree at
  most 3 and $X \subseteq \vertices{T}$. Then~$T$ has
  $\floor{\frac{\card{X}}{2k-1}}-1$ vertex-disjoint subtrees each
  containing at least~$k$ vertices of~$X$. 
\end{prop}

\begin{prop}[\cite{BodlaenderLTT97onin}] 
\label{p:twk2r}
  For any integer~$r \geq 1$ and any graph~$G$,
  \[
  G \notcontain K_{2,r} \Rightarrow \tw(G) < 2r-2
  \]
\end{prop}

\begin{prop}[\cite{JGT:Stiebitz}] 
\label{p:partition}
  For any integer~$k \geq 1$ and any graph~$G$, there exist sets
  $V_1,\dots, V_k$ partitioning $\vertices{G}$ (\ie~$\sqcup_{i \in
    \intv{1}{k}} V_i = \vertices{G}$) such that
\[
\forall i \in \intv{1}{k}, \forall u \in V_i,\ \dg_{V_i}(v) \geq \frac{\dg_G(v)}{k} - 1
\]
In particular, if $\mindeg(G) \geq p$ then $\forall i \in \intv{1}{k},\ \mindeg(\induced{G}{V_i}) \geq \frac{p}{k} -1$
\end{prop}

\begin{prop}[Erd\H{o}s--Szekeres Theorem, \cite{ErdosSzekeres}]\label{p:es} 
 Let~$k$ and~$\ell$ be two strictly positive integers. Then any
 sequence of~${(\ell-1)(k-1) + 1}$ distinct integers contains either
 an increasing subsequence of length~$k$ or a decreasing subsequence
 of length~$\ell$.
\end{prop}

\begin{prop}[\cite{Kostochka82}, \cite{Thomason01}, \cite{DIE10B} (7.2.3)]\label{p:kost}
There is a real constant $c$ such that every graph of average degree
more than a function $c(t) = (c + o(1)) t \sqrt{\log t}$ contains
$K_t$ as minor. According to \cite{Kostochka82}, $c(t) < 648 \cdot t
\sqrt{\log t}.$
\end{prop}

\section{Excluding packings of planar graphs}
\label{sec:results}

Theorems~\ref{first} and~\ref{second} follow combining the two
following results with the machinery introduced
in~\cite{FominST11stre} (see also~\cite{FominLMPS13quad}). They have
independent interest as they detect cases of Theorem~\ref{classic}
where $g$ depends polynomially on $k$.

\begin{theorem} 
\label{th1} Let $H$ be a graph of pathwidth at most 2 on $r > 5$
vertices. If $G$ does not contain $k$ disjoint copies of $H$ as minors then
$\tw(G)\leq 2^{O(r^2)}\cdot k^{2}\cdot \log 2k$.
\end{theorem}

\begin{theorem} 
\label{th2} For every positive integer $r$, if $G$ does not contain
$k$ disjoint copies of $K_{2,r}$ as a minors then $\tw(G) < 20k^2r^2 -
8 k^2r + 2r - 1$.
\end{theorem}

\subsection{Auxiliary results}
\label{sec:lemmas}

\begin{lemma} 
\label{l:good_mesh}
Let~$G$ be a graph and let $p \geq q \geq 1$ be integers. If ${\tw(G) \geq 5 p q - 2 q + 2p - 1}$, then there exist~$2q$ disjoint sets $X_1, \dots, X_{2q}$ of~$\vertices{G}$ and a set~$\mathcal{P}$ of~$pq$ disjoint paths in~$G$ of length at least 2 and such that
\begin{enumerate}[(i)]
\item \label{l:good_mesh:i}$\forall i \in \intv{1}{2q}$, $X_i$ is of size~$p$ and is connected in~$G$ by a tree~$T_i$ using the elements of some set~${A \subseteq \vertices{G}}$;
\item \label{l:good_mesh:ii} any path in~$\mathcal{P}$ has one of its ends in some~$X_i$ with $i \in \intv{1}{q}$, the other end in some $X_j$ with $j \in \intv{q + 1}{2 q}$ and its internal vertices are in none of the~$X_l$, for all $l \in \intv{1}{2q}$, nor in~$A$.
\item \label{l:good_mesh:iii} $\forall i,j \in \intv{1}{2k},\ i \neq j \Rightarrow T_i \cap T_j = \emptyset$
\end{enumerate}
\end{lemma}

\begin{proof}
Let~$G$ be a graph, $p \geq q \geq 1$ two integers and assume that ${\tw(G) \geq 5 p q - 2 q +2p - 1}$.
According to Proposition~\ref{p:mesh},~$G$ contains a~$(pq)$-mesh of order~$(2p - 1)(2q + 1)$. Let~$(A,B)$ be this mesh, $X = A \cap B$ and let $T$ be the tree related to~$A$. By definition of a mesh, $T$ is a tree of maximum degree 3 and~$X \subseteq \vertices{T}$.

Using Proposition~\ref{p:tree_cut}, there exist $\floor{\frac{|X|}{2p - 1}} - 1 = 2q$ disjoint subtrees $T_1,\dots T_{2q}$ of $\vertices{T}$ such that for all $i \in \intv{1}{2q},\ \card{\vertices{T_i} \cap X} \geq p$. For all $i \in \intv{1}{2q}$, let~$X_i$ be a subset of~$\vertices{T_i} \cap X$ such that~$\card{X_i} = p$.

The set~$X$ is externally~$(pq)$-connected in~$B$ (by definition of a mesh), \ie~any two subsets of~$X$ of size~$pq$ are linked by~$pq$ disjoint paths whose internally vertices are in~$B$. Thus, the sets $Z_1 = \bigcup_{i \in \intv{1}{q}} X_i$ and $Z_2 = \bigcup_{i \in \intv{q+ 1}{2q}} X_i$ (whose each is of size~$pq$) are externally connected in~$B$. Let~$\mathcal{P}$ be these~$pq$ paths between~$Z_1$ and~$Z_2$.
We now check the conditions (\ref{l:good_mesh:i}), (\ref{l:good_mesh:ii}) and (\ref{l:good_mesh:iii}) on $\{X_i\}_{i \in \intv{1}{2q}}$ and~$\mathcal{P}$.

\begin{enumerate}
\item[(\ref{l:good_mesh:i})] by definition of $\{X_i\}_{i \in \intv{1}{2q}}$, for all $i \in \intv{1}{2q}$, $\card{X_i} = p$ and~$X_i$ belongs to~$\vertices{T_i}$, therefore~$X_i$ is connected in~$G$ by the tree~$T_i$;
\item[(\ref{l:good_mesh:ii})] $\mathcal{P}$ contains disjoint paths
  such that
  \begin{itemize}
  \item they do not use elements of~$A$ (by definition);
  \item they are external to~$Z_1$ and~$Z_2$ (\ie~none of their
    internal vertices belongs to~$X_i$, for all~$i \in \intv{1}{2q}$);
  \item any~$p \in \mathcal{P}$ links~$Z_1$ to~$Z_2$, thus~$p$ have
    one end in~$Z_1$ and the other end in~$Z_2$, put another way~$p$
    have one end in some~$X_i$ for $i \in \intv{1}{2q}$ and the other
    end in some~$X_j$ for some~$j \in \intv{q+1}{2q}$.
  \end{itemize}
\item[(\ref{l:good_mesh:iii})] by definition the~$T_i$'s are all disjoint.
\end{enumerate}

The sets $\{X_i\}_{i \in \intv{1}{2q}}$ satisfies the properties (\ref{l:good_mesh:i}), (\ref{l:good_mesh:ii}) and (\ref{l:good_mesh:iii}) so we found these sets we were looking for.
\end{proof}

\begin{lemma} 
  \label{l:smalldeg}
  For any integer~$a \geq 1$ and for any graph $G$, $\vertices{G}$ contains more than $(1-\frac{1}{a}) \card{\vertices{G}}$ vertices of degree strictly less than~$2a\dege{G}$. In particular, $\vertices{G}$ contains at least $\frac{\card{\vertices{G}}}{2}$ vertices of degree strictly less than~$\dege(G)$.
\end{lemma}

\begin{proof}
Let~$a \geq 1$ be an integer and let~$G$ be a graph.

Let~$n_h$ be the number of vertices of~$G$ with degree at least $h = 2 a \times \dege(G)$, \ie~$n_h = \card{\{v \in \vertices{G},\ \dg(v) \geq h\}}$ and~$n_{-h}$ the number of vertices of degree strictly less than~$h$, \ie~$n_{-h} = \card{\vertices{G}} - n_h$.
Clearly, there is at least~$\half h n_h$ edges incident the~$n_h$ vertices of degree at least~$h$. We thus have:

\begin{align*}
  \half h n_h &\leq \card{\edges{G}} & \text{(because there may be other edges)}\\
  & \leq \half \sum_{v \in \vertices{G}} \dg(v) & \text{(Handshaking lemma)}\\
  \frac{h n_h}{\card{\vertices{G}}} & \leq \frac{\sum_{v \in \vertices{G}} \dg(v)}{\card{\vertices{G}}} &\\
  & < 2 \dege(G) & \paren{\text{because}\ \frac{\sum_{v \in \vertices{G}} \dg(v)}{\card{\vertices{G}}} = \ad(G) < 2 \dege(G)}\\
  n_h & < \card{\vertices{G}} \frac{2 \dege(G)}{h} & \\
  n_{-h} & > \card{\vertices{G}} \paren{1 - \frac{2 \dege(G) }{h}} &\\
  & > \card{\vertices{G}} \paren{1 - \frac{1}{a}} & \text{(by replacing~$h$ by its value)}
\end{align*}
Finally, we found that~$G$ contains more than $\card{\vertices{G}} \paren{1 - \frac{1}{a}}$ vertices of degree strictly less than $2 a \times \dege(G)$, what we wanted to prove.
\end{proof}

\begin{lemma} 
\label{l:big-degec}
  Let~$k, r$ be two positive integers and~$G$ a graph such that~${\degec(G) \geq 2kr}$. Then~$G$ contains~$k$ disjoint copies of~$\K{2,r}$ as minors.
\end{lemma}

\begin{proof}
Let~$k, r$ be two positive integers and~$G$ a graph of contraction degeneracy at least~$2kr$. Then~$G$ has a minor~$G'$ such that $\mindeg(G') \geq 2kr$.

According to Proposition~\ref{p:partition}, there is a partition $\mathcal{V} = \{V_1,\dots, V_k\}$ of~$\vertices{G'}$ such that
\[
  \forall V_i \in \mathcal{V},\ \mindeg(\induced{G'}{V_i}) \geq \frac{2kr}{k} - 1 = 2r - 1
\]

The minimum degree of a graph is a lower bound for its treewidth, then any $V_i \in \mathcal{V}$ has treewidth at least~$2r-1$, and thus by Proposition~\ref{p:twk2r} $V_i$ contains~$\K{2,r}$ as a minor.
$\mathcal{V}$ is a partition of size~$k$ of~$\vertices{G'}$ and each element of~$\mathcal{V}$ contains~$\K{2,r}$ as a minor consequently $G'$ contains~$k$ disjoint copies of~$\K{2,r}$ as minors. As~$G'$ is a minor of~$G$, $G$ contains~$k$ disjoint copies of~$\K{2,r}$ as minors, what we wanted to show.

\end{proof}

\begin{lemma} 
\label{l:path-tree}
Let~$T$ be a ternary tree and $X = \{v \in \vertices{T},\ \dg_T(v)
\leq 2\}$. Then
\begin{enumerate}[(i)]
\item for any path $P$ on $l$ vertices in~$T$, $T$ has a partition
  $\mathcal{M}$ such that
  \begin{enumerate}[a)]
  \item every vertex of $P$ belongs to a different element of
    $\mathcal{M};$
  \item every element of $\mathcal{M}$ contains an element of $X.$
  \end{enumerate}
  \label{l9:i}
\item $T$ has diameter at least $2 \log \frac{2}{3}\card{X}$\label{l9:ii}.
\end{enumerate}
\end{lemma}

\begin{proof}[Proof of~(\ref{l9:i})]
Let $T,$ $X$, $P$ be as in the statement of the lemma.
For every $u \in \vertices{P},$ we set $M_u$ as the set of vertices of
the connected component $G \setminus (P \setminus \{u\})$ that
contains $u.$ Let $\mathcal{M} = \{M_u\}_{u \in P}.$ Clearly, for all
$u,v \in \vertices{P}$, if $u \neq v$ then $M_u \cap M_v = \emptyset.$
Also, since $T$ is connected, there is no vertex of $\vertices{T}$
that is not in an element of $\mathcal.$ Therefore $\mathcal{M}$ is a
partition of $\vertices{T}$. By definition, for every $u \in
\vertices{P},\ u \in M_u$. Besides, every element $M$ of $\mathcal{M}$
contains either exactly one element, which is necessarily a vertex of
degree 2 in $T$, or more than one element ad in this case it induces
in $G$ a tree whose leaves are also leaves of $G$. In both cases $M$
contains an element of $X$ as required.
\end{proof}

\begin{proof}[Proof of~(\ref{l9:ii})]
Let $P = p_0 \dots p_k$ be a longest path in $T$. In order to be able
to use the notions of height and of child, we root $T$ at node
$n_{\floor{\frac{k}{2}}}$ (which is clearly not a leave).

We prove the proposition for the case where~$T$ has no vertices of
degree two. If this is not the case, we can just add a leaf as child
of every vertex of degree two. As these vertices have an other child,
there is at least one longest path that use none of the new vertices.

Let~$\ell = \card{X}.$ By contradiction, assume that $k < 2 \log \frac{2}{3}\ell$.

Let $T'$ be the full ternary tree of height $\ceil{\frac{k'}{2}}.$ As
$T'$ is complete, it has $3 \cdot 2^{\ceil{\frac{k}{2}}-1}$
leaves. The tree $T'$ clearly contains $T$ as subgraph because they have same
height, thus $T'$ has at most as much leaves as $T,$ \ie~$l \leq 3 \cdot
2^{\ceil{\frac{k}{2}}-1}$.  If we use our first assumption, we get:
\begin{align*}
  l & \leq 3 \cdot 2^{\ceil{\frac{k}{2}}-1}\\
   & < 3 \cdot 2^{\ceil{\log \frac{2}{3}\ell}-1}\\
   l < l
\end{align*}
We obtain a contradiction, thus our assumption $k < 2 \log
\frac{2}{3}\ell$ was false: $T$ has diameter at least $2 \log
\frac{2}{3}\card{X}.$
\end{proof}

\begin{lemma} 
\label{l:independent} 
Let~$k,r$ be two positive integers and $G = ((V_1,V_2),E)$ a bipartite multigraph such that
\begin{align*}
\card{V_1} &= \card{V_2} \geq 4k^2r\\
\forall v \in \vertices{G},\ \multidg(v) &= 2kr^2\\
\dege(G) & < 2kr
\end{align*}
Then~$G$ has at least~$k$ (vertex-)disjoint multiedges of multiplicity at least~$r$.
\end{lemma}

\begin{proof}
Let~$G$ be a graph that fill the conditions of the lemma.
For $(u,v) \in \edges{G}$, let~$\mult(u,v)$ denote the multiplicity of the edge~$(u,v)$.
According to lemma~\ref{l:smalldeg}, $G$ contains at least $\half \vertices{G} \geq 4k^2r$ vertices of degree strictly less than~$\dege(G) < 2kr$. Then, one of~$V_1,V_2$ contains at least~$2k^2r$ such vertices. We assume without loss of generality that this is~$V_1$. Let~$L$ be a subset of~$V_1$ of size~$2k^2r$ containing vertices of degree strictly less than~$2kr$. For all~$v \in L$, $v$ has degree less than~$2kr$ (by definition of~$L$) and multidegree~$2kr^2$ (by initial assumption) so there is a least one~$u \in V_2$ such that~$\mult(u,v) \geq r$.

We now define an auxiliary function. Let~$f : L \to V_2$ a function such that $\forall v \in L,\ \mult(v, f(v)) \geq r$. According to the previous remark, such a function exists.
For all~$u \in f(L)$, the multidegree of~$u$ is by assumption~$2kr^2$ thus~$u$ cannot be the image of more than $\frac{\multidg(u)}{r} = 2kr$ elements of~$L$. Consequently,~$f(L)$ has size at least $\frac{\card{L}}{2kr} \geq k$. Remark that for all $u_1,u_2 \in f(L)$ with $u_1 \neq u_2$, the preimages of~$u_1$ and~$u_2$ are disjoint.

We finally show~$k$ disjoint multiedges of multiplicity at least~$r$ in~$G$. Choose~$k$ distinct elements $u_1, \dots, u_k$ of~$f(L)$ and for all $i \in \intv{1}{k}$ let~$v_i$ be an element of~$L$ in the preimage of~$u_i$ (\ie~such that~$f(v_i) = u_i$). As said before, the preimages of distinct elements of~$f(L)$ are distinct so the~$v_i$'s are all distinct. By definition $\forall i \in \intv{1}{k}, f(v_i) = u_i$ so there is an edge of multiplicity~$r$ between~$u_i$ and~$v_i$ in~$G$.
Therefore, $\{(v_i, u_i)\}_{i \in \intv{1}{k}}$ is the set of edges we were looking for.

\end{proof}

In~\cite{2013arXiv1305.7112R} we prove the following lemma.
\begin{lemma}[\cite{2013arXiv1305.7112R}]
\label{l:pw2sg}
  For all graph $G$, if $n = \card{\vertices{G}}$, then $\pw(G) \leq 2 \Rightarrow G \lminor \sg{n}$.
\end{lemma}

\begin{lemma}\label{l:mesh-plus}
 For all positive integers~$p,q$ and all graph~$G$, if $\tw(G) \geq 20p^2q^2 - 8 p^2q + 2q - 1$ and $\degec(G) < 2pq$ then~$G$ contains~$2p$ disjoint subsets $X_1, \dots, X_{2p}$ of $\vertices{G}$ and a set $\mathcal{P}$ of~$pq$ disjoint paths of length at least 2 in~$G$ such that

\begin{enumerate}[(i)]
\item \label{l:pairs:i}$\forall i \in \intv{1}{2p}$, $X_i$ is of size~$q$ and is connected in~$G$ by a tree~$T_i$ using the elements of some set~$A \subseteq \vertices{G}$;
\item \label{l:pairs:ii} any path in~$\mathcal{P}$ has one of its ends in some~$X_i$ with $i \in \intv{1}{p}$, the other end in~$X_{2i}$ with $j \in \intv{q + 1}{2 p}$ and its internal vertices are in none of the~$X_l$, for all $l \in \intv{1}{2p}$, nor in~$A$;
\item \label{l:pairs:iii}$\forall i,j \in \intv{1}{2p},\ i \neq j \Rightarrow T_i \cap T_j = \emptyset$.
\end{enumerate}
\end{lemma}

\begin{proof}
  According to lemma~\ref{l:good_mesh},~$G$ contains~$8p^2q$ disjoint sets $Y_1, \dots, Y_{8p^2q}$ of $\vertices{G}$ and a set $\mathcal{P}$ of~$4p^2q^2$ disjoint paths in~$G$ of length at least 2 and such that
\begin{enumerate}[(i)]
\item $\forall i \in \intv{1}{8p^2q}$, $Y_i$ is of size~$q$ and is connected in~$G$ by a tree~$T_i$ using the elements of some set~$A \subseteq \vertices{G}$;
\item any path in $\mathcal{P}$ has one of its ends in some~$Y_i$ with $i \in \intv{1}{4p^2q}$, the other end in some~$Y_j$ with $j \in \intv{4p^2q + 1}{8p^2q}$ and its internal vertices are in none of the~$Y_l$, for all $l \in \intv{1}{8p^2q}$, nor in~$A$;
\item $\forall i,j \in \intv{1}{8p^2q},\ i \neq j \Rightarrow T_i \cap T_j = \emptyset$.
\end{enumerate}

Let us consider the bipartite multigraph~$H$ defined by
\begin{itemize}
\item $\vertices{H} = \{Y_i\}_{i \in \intv{1}{8p^2q}}$ ;
\item for all~$n$ integer and $i,j \in \intv{1}{8p^2q}$ there is an edge of multiplicity~$m$ between the two vertices~$Y_i$ and~$Y_j$ iff there is exactly~$m$ paths from a vertex of~$Y_i$ to a vertex of~$Y_j$ in~$P$.
\end{itemize}

Clearly,~$H$ is a minor of~$G$. Consequently~$2pq > \degec(G) \geq \degec(H) \geq \dege(H)$.

The three conditions required on~$H$ by lemma~\ref{l:independent} are filled, so~$H$ contains~$p$ disjoint multiedges of multiplicity~$q$.

By construction of~$H$, having an edge of multiplicity~$m$ in~$H$ is
equivalent to having~$m$ distinct paths in~$P$ between two sets~$Y_i$
and~$Y_j$, then having~$p$ disjoint multiedges of multiplicity~$q$
in~$H$ is equivalent to having~$p$ disjoint pairs $(X_i, X_{2i})_{i
  \in \intv{1}{p}}$ of elements of $\{Y_i\}_{i\in \intv{1}{4p^2q}}$
and a set $P$ of $pq$ paths that contains~$q$ paths that
links the two elements of each of the $p$ pairs. The set $\{X_i\}_{i
  \in \intv{1}{2p}}$ is thus the one we were looking for.
\end{proof}

\subsection{Proof of Theorem~\ref{th1}}
\label{sec:th1}

\begin{proof}[Proof of theorem~\ref{th1}.]

We prove the contrapositive. Let~$k$ be a integer,~$H$ a graph on~$r>5$
vertices and of pathwidth at most 2 and~$G$ a graph.
From Proposition~\ref{l:pw2sg},~$H\lminor\sg{r}$. If we show that~$G$
contains~$k$ disjoint copies of~$\sg{r}$ as minors then we are done.
Let $g \colon \N \rightarrow \N$ such that
\[
g(k,r) = k^2 \log 2k \paren{180 \cdot 2^{r(r-2)} - 24 \cdot 2^{\half r(r-2)}} + 6 \cdot 2^{\half r(r-2)} - 1
\]
We prove the following statement: for all graph $G$,
$\tw(G) \geq g(k,r)$ implies that~$G\contains k\cdot\sg{r}$.
Let $k$ and $r>5$ be two positive integers and assume that~$\tw(G) \geq
g(k,r)$.

\noindent \textbf{First case:}~$\degec(G) \geq c \cdot 3 rk \sqrt{
\log 3rk}$.

By definition of the contraction degeneracy, there is a graph~$G'$
minor of~$G$ and such that ${\mindeg(G') \geq c \cdot 3 r k \sqrt{\log
3 r k}}$. The average degree is at least the minimum degree, so
${\ad(G') \geq c \cdot 3 r k \sqrt{\log 3 r k}}$. According to
Proposition~\ref{p:kost}, $G'$ contains~$K_{3kr}$ as minor.

The graph~$\sg{r}$ have~$3r$ vertices, therefore~$K_{3kr}$ contains~$k
\cdot \sg{r}$ as minor. We then have $k \cdot \sg{r} \lminor K_{3kr}$,
$K_{3kr} \lminor G'$ and $G' \lminor G$, therefore by transitivity of
the minor relation, $G$ contains $k\cdot \sg{r}$ as minor, what we
wanted to show.

\noindent \textbf{Second case:}~$\degec(G) < c \cdot 3 rk \sqrt{\log
3rk}$.

Observe that $c \cdot 3 rk \sqrt{\log 3rk}< c\cdot 3r
\sqrt{\log 6r} \cdot k\sqrt{\log 2k}$. Let $k_0 = k \sqrt{\log 2k}$ and $r_0 = 3 \cdot 2^{\frac{r(r-2)}{2}},$ and
remark that $k_0 \geq k$ and, $r_0 \geq c \cdot 3 r \sqrt{\log 6r}$ (remember
that $c \leq 648$ and $r > 5$). With these notations, we have
${\degec(G) < 2 k_0 r_0}$. We will show that $G\contains k_0 \cdot
K_{2,r}$ from which yields that $G\contains k \cdot K_{2,r}$. By
assumption, $\tw(G) \geq g(k,r)$. Therefore, by
Lemma~\ref{l:mesh-plus} (applied for ${p := k_0}$ and ${q := r_0}$),
$G$ contains~$2k_0$ subsets $X_1, \dots, X_{2k_0}$ of~$\vertices{G}$
and a set $\mathcal{P}$ of $k_0r_0 = 3 k_0 \cdot 2^{\frac{r(r-2)}{2}}$
disjoint paths of length at least 2 in~$G$ such that

\begin{enumerate}[(i)]
\item $\forall i \in \intv{1}{2k_0}$, $X_i$ is of size~$r_0 = 3 \cdot
2^{\frac{r(r-2)}{2}}$ and is connected in~$G$ by a tree~$T_i$ using
the elements of some set~$A \subseteq \vertices{G}$;
\item any path in~$\mathcal{P}$ has one of its ends in some~$X_i$ with
$i \in \intv{1}{k_0}$, the other end in~$X_{2i}$ and its internal
vertices are in none of the~$X_l$, for all~$l \in \intv{1}{2k_0}$, nor
in~$A$;
\item $\forall i,j \in \intv{1}{2k_0},\ i \neq j \Rightarrow T_i \cap
T_j = \emptyset$.
\end{enumerate}

We assume that for all $i \in \intv{1}{2k_0}$, $X_i = \{v \in
\vertices{T_i},\ \dg_T(v) \leq 2\}$. It is easy to come down to this
case by considering the minor of $G$ obtained after deleting in~$T_i$
the leaves that are not in~$X_i$ and contracting one edge meeting a
vertex of degree 2 which is not in~$X$ while such a vertex exists.

As $T_{i}$ is a ternary tree, one can easily prove that for all $i
\in \intv{1}{2k_0}$, $T_i$ contains a path containing $2 \log
\frac{2}{3}\card{X_i} = (r - 1)^2 + 1$ vertices of~$X_i$. Let us
call~$P_i$ such a path whose two ends are in~$X_i$. Let us consider
now the paths $\{P_i\}_{i \in \intv{1}{2k_0}}$ and the paths that
link the elements of different~$P_i$'s.
For each path $i \in \intv{1}{2k_0}$, we choose in~$P_i$ one end
vertex (remember that both are in~$X_i$) that we name~$p_{i,0}$. We
follow~$P_i$ from this vertex and we denote the other vertices of~$P_i
\cap X_i$ by $p_{i,1}, p_{i,1}, \dots, p_{i,(r-1)^2}$ in this order.
The \emph{corresponding vertex} of some vertex~$p_{i,j}$ of~$P_i \cap
X_i$ (for~$i \in \intv{1}{k_0}$) is defined as the vertex of~$P_{2i}
\cap X_{2i}$ to which $p_{i,j}$ is linked to by a path
of~$\mathcal{P}$.

As said before, the sets $\{P_i \cap X_i\}_{i \in \intv{1}{2k_0}}$ are
of size~$(r-1)^2 + 1$. According to Proposition~\ref{p:es}, one can find for
all $i \in \intv{1}{k_0}$ a subsequence of length~$r$ in
$p_{i,0},p_{i,1}, \dots, p_{i,(r-1)^2}$, such that the corresponding
vertices in~$X_{2i}$ of this sequence are either in the same order
(with respect to the subscripts of the names of the vertices), or in
reverse order. For all $i \in \intv{1}{k_0}$, this subsequence, its
corresponding vertices and the vertices of the paths that link them
together forms a~$\sg{r}$ model. We have thus~$k_0$ models of~$\sg{r}$
in~$G$, that gives us $k$ disjoint models of~$\sg{r}$ in~$G$ (since $k
\leq k_0$).

We showed that for all~$k$ and~$r > 5$ positive integers, if a graph~$G$ has
$\tw(G) \geq g(k,r)$, then ${G\contains k\cdot \sg{r}}$. For every
graph~$H$ on $r$ vertices and of pathwidth at most~$2$, $H$ is a minor
of the subdivided grid~$\sg{r}$ (Proposition~\ref{l:pw2sg}). Consequently,
if~$G$ has treewidth at least $g(k,r)$, then~$G$ contains~$k$
disjoint copies of~$H$ and we are done.
\end{proof}

\subsection{Proof of Theorem~\ref{th2}}
\label{sec:th2}

\begin{proof}[Proof of theorem~\ref{th2}.]
We prove the contrapositive. Let~$k$ and~$r$ be two positive integers and~$G$ a graph such that $\tw(G) \geq 20k^2r^2 - 8 k^2r + 2k - 1$. We want to show that~$G$ contains~$k$ disjoint copies of~$\K{2,r}$.

\noindent \textbf{First case:}~$\degec(G) \geq 2kr$

According to lemma~\ref{l:big-degec},~$G$ contains~$k$ disjoint copies of~$\K{2,r}$, what we wanted to show.

\noindent \textbf{Second case:}~$\degec(G) < 2kr$

According to lemma~\ref{l:mesh-plus}, there exist~$2k$ disjoint subsets $X_1, \dots, X_{2k}$ of~$\vertices{G}$ and a set~$\mathcal{P}$ of disjoint paths of length at least 2 such that

\begin{enumerate}[(i)]
\item $\forall i \in \intv{1}{2k}$, $X_i$ is of size~$r$ and is connected in~$G$ by a tree~$T_i$ using the elements of some set $A \subseteq \vertices{G}$;
\item any path in~$\mathcal{P}$ has one of its ends in some~$X_i$ with $i \in \intv{1}{k}$, the other end in~$X_{2i}$ with $j \in \intv{q + 1}{2 k}$ and its internal vertices are in none of the~$X_l$, for all $l \in \intv{1}{2k}$, nor in~$A$;
\item $\forall i,j \in \intv{1}{2k},\ i \neq j \Rightarrow T_i \cap T_j = \emptyset$.
\end{enumerate}

We then perform the following operations on~$G$.
\begin{enumerate}
\item for all $i \in \intv{1}{2k}$, we contract the set~$X_i$ to a single vertex~$x_i$ (this is possible because~$X_i$ is connected by the tree~$T_i$);
\item for all path~$p \in \mathcal{P}$, we contract some edges of~$p$ until it have length exactly 2.
\end{enumerate}
Because it has been obtained by contraction of edges, the graph~$G'$ we get by these operations is a minor of~$G$. This new graph has the following properties.
\begin{enumerate}
\item for all $i \in \intv{1}{k}$, the vertex~$x_i$ is linked to the vertex~$x_{2i}$ by~$r$ disjoint paths of length 2;
\item for all $i,j \in \intv{1}{k}\ i \neq j \Rightarrow x_i \neq x_j$ because the trees~$T_i$ and~$T_j$ contracted to obtain~$x_i$ and~$x_j$ are disjoint.
\end{enumerate}
Remark that for all $i \in \intv{1}{k}$, the subgraph of~$G'$ induced
by the vertices~$x_i$,~$x_{2i}$ and the~$r$ middle vertices of the
paths of length 2 that links~$x_i$ and~$x_{2i}$ is the
graph~$\K{2,r}$. We consequently found~$k$ disjoint copies
of~$\K{2,r}$ in a minor of~$G$, so~$G$ contains~$k \times \K{2,r}$ as
minor, what we wanted to prove.

\end{proof}

\section{From planar graph exclusion to \texorpdfstring{Erd\H{o}s--P\'osa}{Erdos-Posa} Property}
\label{sec:excl_ep}

In the section, we adapt to our needs the technique introduced
in~\cite{FominST11stre} (and also used in~\cite{FominLMPS13quad}) to
translate a bound on the treewidth of a graph that does not contain a
planar graph as minor to a gap for the {Erdős--Pósa}
Property. We need two lemmata and a theorem in order to prove Theorems~\ref{first}
and~\ref{second}.

\begin{lemma}[adapted from~\cite{FominST11stre}]\label{l:pack_sep}
Let $H$ be a connected planar graph. Every graph $G$ of treewidth $w$ such that
$\pack{H}(G) = k$ has a separation $(A,B)$ of order at most $w + 1$ satisfying
$\pack{H}(\induced{G}{A \setminus B}) \leq \floor{\third {2k}}$ and $A \cup B = \vertices{G}$.
\end{lemma}

\begin{proof}
  Let~$H$ be a connected planar graph, $G$ be a graph of treewidth~$w$ such that
  $\pack{H}(G) = k$ and~$(T,V)$ be a nice optimal tree decomposition
  of~$G$. For every $t \in \vertices{T}$, we denote by~$G_t$ the
  subgraph of~$G$ equal to $\induced{G}{\paren{\cup_{u \in
        \desc{T}(t)} V_u} \setminus V_t}.$ We consider the function $p
  : \vertices{T} \to \N$ defined by $\forall t \in \vertices{T},\ p(t)
  = \pack{H}(G_t)$. Let us now state some remarks about the function
  $p$.

  \begin{remark}\label{r:non_dec}
    For every two vertices $u,v \in \vertices{T}$, if $v \in \desc{T}(u)$
    then $p$ is non-decreasing along the (unique) path of $T$ from $v$
    to $u$. To see this, it suffices to remark that if $t\in
    \vertices{T}$ has child $t'$, then $G_{t} \supseteq G_{t'}$ (what
    implies that $G_{v} \supseteq G_u$).

    In particular, $p$ is non-decreasing along the path from every vertex
    of $T$ to the root of $T$.
  \end{remark}

  \begin{remark}\label{r:claims}
    As $T$ is a nice decomposition of $G$, its vertices can be of four
    different kinds:
    \begin{itemize}

    \item Base node $t:$ $p(t) = 0$ because as $t$ has no descendant,
      $G_t = \emptyset$;

    \item Introduce node $t$ with child $t':$ as the unique element of
      $V_t \setminus V_{t'}$ cannot appear in the elements of $\desc{T}(t')$
      (by definition of a tree decomposition), $G_t = G_{t'}$ and
      then $p(t) = p(t');$

    \item Forget node $t$ with child $t':$ in this case, the unique
      element of $G_{t} \setminus G_{t'}$ may be part of at most one
      model of $H$ in $G_t$ (because we want vertex-disjoint models)
      therefore either $p(t) = p(t')$ or $p(t) = p(t') + 1;$

    \item Join node $t$ with children $t_1$ and $t_2:$ the graphs
      $G_{t_1}$ and $G_{t_2}$ are disjoint and $G_{t} = G_{t_1} \cup
      G_{t_2}$. As $H$ is connected, there is no model of $H$ in $G_t$
      that is simultaneously in $G_{t_1}$ and in $G_{t_2}$,
      consequently $p(t) = p(t_1) + p(t_2)$.
    \end{itemize}
  \end{remark}

  Let $t$ be a vertex of $T$ such that $p(t) > \frac{2}{3}k$ and for
  every child $t'$ of $t$, $p(t') \leq \frac{2}{3} k.$

  We make some claims about this vertex $t$:
  \begin{enumerate}[(1)]
  \item \label{claim1} such $t$ exists;
  \item \label{claim2} $t$ is unique;
  \item \label{claim3} $t$ is either a forget node or a join node.
  \end{enumerate}

  \begin{proofclaim}[of Claim~(\ref{claim1})]
    The value of $p$ on the root $r$ of $T$ is $k$ (because
    $G_r = G$) and the value of $p$ on every base nodes $b$ is 0
    (because $G_b$ is the empty graph). As
    $p$ is non decreasing on a path from a base node to the root
    (Remark~\ref{r:non_dec}), a vertex such $t$ exists.
  \end{proofclaim}
  
  \begin{proofclaim}[of Claim~(\ref{claim2})]
    To show that $t$ is unique, we assume by contradiction that there is
    another $t'\in V(T)$ with $t' \neq t$ and $p(t') > \frac{2}{3}k$
    and for every child $t''$ of $t$, $p(t'') \leq \frac{2}{3}
    k$. Three cases can occur:
    \begin{itemize}
    \item either $t'$ is a descendant of $t$. However, $p$ is non
      decreasing on a path from a vertex to the root
      (Remark~\ref{r:non_dec}) and $p(t') \geq
      \frac{2}{3}k$ whereas the value of $p$ for each child of $t$ is
      at most $\frac{2}{3}k$ (by definition of $t$): this is a contradiction.
    \item or $t$ is a descendant of $t'$ and the same argument
      applies (symmetric situation).
    \item or $t$ and $t'$ are not in the above situations. Let $v \in
      \vertices{T} \setminus \{t, t'\}$
      be the least common ancestor of $t$ and $t'$. As $p$ is non
      decreasing along any path from a vertex to the root
      (Remark~\ref{r:non_dec}), the child $v_t$ (resp.\ $v_{t'}$) of
      $v$ whose $t$ (resp.\ $t'$) is descendant of should be such
      $p(v_t) > \frac{2}{3} k$ (resp.\ $p(v_{t'}) > \frac{2}{3} k).$
      By definition of $v$, we have $v_t \neq v_{t'}.$ As $v$ is a join
      node, $p(v) = p(v_t) + p(v_{t'}) > \frac{4}{3} k$, what is
      impossible.
    \end{itemize}
  \end{proofclaim}

  \begin{proofclaim}[of Claim~(\ref{claim3})]  
    By definition the value of $p$ on $t$ is strictly positive and
    different from the value of $p$ on every child of $t$. As this
    cannot occur with introduce nodes (where $p$ take on $t$ the same
    value it takes on the child of $t$) nor on base nodes (where $p$
    is null), $t$ is either a join node or a forget node.

  \end{proofclaim}

  We now present a separation $(A,B)$ of order at most $w + 1$ in $G$.
  
  \noindent \textbf{Case 1:} $t$ is a forget node with $t'$ as child.

  Let $A = \vertices{G_t} \cup V_t$ and $B = \vertices{G} \setminus
  \vertices{G_t}$.
  
  \noindent \textbf{Case 2:} $t$ is a join node with $t_1,\ t_2$ as
  children.

  By definition of $t$ we have $p(t) \geq \third {2k}$. As $p(t) = p(t_1) +
  p(t_2)$ (according to Remark~\ref{r:claims}) there is a $i \in
  \{1,2\}$ such that $p(t_i) \geq \third k$.
  Let $A = \vertices{G_{t_i}} \cup V_t$ and $B = \vertices{G} \setminus
  \vertices{G_{t_i}} $

  In both cases, we have
  \begin{enumerate}[(i)]
  \item \label{i:no_edge} there is no edge between $A \setminus B$
    and $B \setminus A$ therefore $(A,B)$ is a separation;

  \item $\card{A \cap B} \leq w + 1$ because
    $A \cap B = V_t$ and $V_t$ is a bag in an optimal tree
    decomposition of $G$ which have treewidth $w$, thus $(A,B)$ is a
    separation of order at most $w + 1;$

  \item $\pack{H}(\induced{G}{A \setminus B}) \leq \frac{2}{3}k$ by
    definition of $A$ and $t;$

  \item $A \cup B = \vertices{G}$ by definition of $B.$
\end{enumerate}
Consequently, the pair $(A,B)$ is a separation of the kind we were
looking for.
\end{proof}

\begin{lemma}[adapted from~\cite{FominST11stre}]\label{l:sep_ep}
  Let $H$ be a connected planar graph, let $\varepsilon >0$ be a real,
  and let $g \colon \N \to \N$ be a function such that $g(n) =
  \Omega(n^{1 + \varepsilon})$. For every integer $k>0$ and graph $G$
  of treewidth less than $g(k)$, if $G$ contains less than $k$
  disjoint models of $H$ then $G$ has a $H$-hitting set of size
  $O(g(k))$.
\end{lemma}

\begin{proof}
  We assume that $\tw(G) < g(k)$ and that $\pack{H}(G) < k$.
  According to Lemma~\ref{l:pack_sep}, there is in $G$ a separation
  $(A,B)$ of order at most $g(k)$ such that
  $\pack{H}(\induced{G}{A \setminus B}) \leq \floor{\third {2k}}$ and
  $A \cup B = \vertices{G}$.

  Remark that as $\{A \setminus B, A \cap B, B \setminus A\}$ is a partition of
  $\vertices{G}$ such that there is no edge between $A \setminus B$
  and $B \setminus A$ (because $(A,B)$ is a separation), every model
  of the connected graph $H$ that use vertices of $A \setminus B$ and
  of $B \setminus A$ also use vertices of $A \cap B$. Consequently we
  have
  \begin{align}\label{e:cover}
  \cover{H}(G) \leq \cover{H}(\induced{G}{A \setminus B}) +
  \cover{H}(\induced{G}{B \setminus A}) + \card{A \cap B}    
  \end{align}

  As $H$ is connected and $A \setminus B$ is disjoint from $B
  \setminus A$, we also have
  \[
  \pack{H}(G) \geq \pack{H}(\induced{G}{A \setminus B}) +
  \pack{H}(\induced{G}{B \setminus A})
  \]
Let $\alpha \in [0,1]$ be a real such that
\begin{align}\label{e:alpha}
  \pack{H}(\induced{G}{A \setminus B}) &\leq \alpha \cdot \pack{H}(G)\\
  \pack{H}(\induced{G}{A \setminus B}) &\leq (1- \alpha) \cdot \pack{H}(G)
\end{align}
We are looking for a function $f$ satisfying the inequality
${\cover{H}(G) \leq f(\pack{H}(G))}$ for every graph $G$ and for every planar
connected graph $H$. A consequence of the grid-exclusion
theorem (see \cite{RobertsonS86GMV} and Theorems 12.4.4 and 12.4.10 of
~\cite{DIE10B}) is that every planar graph has the
Erdős--Pósa Property, thus a function such $f$ exists. We assume
without loss of generality that
\begin{align}\label{e:extra}
f(\pack{H}(G)) \leq \cover{H}(\induced{G}{A \setminus B}) +
  \cover{H}(\induced{G}{B \setminus A}) + \card{A \cap B}  
\end{align}
(to ensure
  this we can choose as value for $f(\pack{H}(G))$ the minimum of the
  right part of the inequality on all graphs $F$ such that
  $\pack{H}(F) = \pack{H}(G)$).

By combining the definition of $f$ with~(\ref{e:extra}),
(\ref{e:cover}) and (\ref{e:alpha}) and using the fact that $(A,B)$
has order at most $g(k)$, we get
\begin{align*}
f(\pack{H}(G))&\leq \cover{H}(\induced{G}{A \setminus B}) +
  \cover{H}(\induced{G}{B \setminus A}) + \card{A \cap B}\\
&\leq f(\pack{H}(\induced{G}{A \setminus B})) +
  f(\pack{H}(\induced{G}{B \setminus A})) + \card{A \cap B}\\
&\leq f(\pack{H}(\induced{G}{A \setminus B})) +
  f(\pack{H}(\induced{G}{B \setminus A})) + g(k)\\
f(\pack{H}(G))&\leq f(\alpha \cdot \pack{H}(G)) + f((1- \alpha) \cdot \pack{H}(G)) + g(k)\\
\end{align*}

By the Akra--Bazzi Theorem~\cite{AkraB98}, the recurrence $h(p) =
h(\alpha p) + h((1-\alpha) p) + g(p)$ where $g(p) = \Omega(p^{1 + \varepsilon})$ is satisfied by
a function $f(p) = O(g(p))$.
Therefore we have $\cover{H}(G) \leq f(k) = O(g(k))$, which means
that $G$ has a $H$-hitting set of size $O(g(k))$, what we wanted to
prove.
\end{proof}

The proofs of Theorems~\ref{first} and~\ref{second} immediately follow
from this theorem combined with lemmata~\ref{th1} and~\ref{th2}. 

\begin{theorem}[adapted from~\cite{FominST11stre}]\label{t:excl_ep}
  Let $H$ be a connected planar graph, let $\varepsilon >0$ be a
  real. Assume that there is a function $g \colon \N \to \N$ such that $g(n) =
  \Omega(n^{1 + \varepsilon})$ and for all graph $G$, for all integer
  $k>0$, $\tw(G) \geq g(k) \Rightarrow G \contains k \cdot H$. Then
  $H$ has the Erdős--Pósa Property with gap $f_H(k) = O(g(k)).$
\end{theorem}

\begin{proof}
  Let $H$, $\varepsilon$ and $g$ be as in the statement of the lemma.
  Let $G$ be a graph.

 \noindent \textbf{Case 1:} $\tw(G) \geq g(k)$

 By definition of $g,$ $G$ contains~$k \cdot H$.

 \noindent \textbf{Case 2:} $\tw(G) < g(k)$

 If $G$ does not contain $k$ disjoint models of $H$, it has a
 $H$-hitting set of size $O(g(k))$ according to Lemma~\ref{l:sep_ep}.

 Consequently, either~$G$ contains $k$ disjoint models of $H$, or $G$
 has a $H$-hitting set of size $O(g(k)),$ in other words: $H$~has the Erdős--Pósa
 Property with gap~$f_H(k) = O(g(k)).$
\end{proof}

\begin{proof}{Proofs of Theorems~\ref{first} and~\ref{second}}
According to Theorem~\ref{th1}, there is a function $f(k)=2^{O(h^2)}\cdot k^2 \cdot \log k$ such that for every graph $H$ on $h$ vertices and of pathwidth at most 2, every graph $G$ of treewidth more than $f(k)$ contains $k$ disjoint copies of $H$. The application of Theorem~\ref{t:excl_ep} immediately yields that the graphs of pathwidth at most 2 have the {Erdős--Pósa} Property with gap at most $f.$

Similarly, since Theorem~\ref{th2} ensure that every graph of treewidth more than some function $g(k,r) = O(k^2r^2)$ contains $k$ disjoint copies of $K_{2,r}$, the application of Theorem~\ref{t:excl_ep} gives that for every integer $r>0,$ the graph $K_{2,r}$ has the {Erdős--Pósa} Property with gap at most $g.$
\end{proof}

\paragraph{Postscript.} Very recently, the general open problem of estimating 
 $f_{H}(k)$ when $H$ is a general planar graph has been  tackled in~\cite{ChekuriC13larg}.
 Moreover, very recently, using the  results of~\cite{LeafS12sube} we were able to improve 
 both Theorems~\ref{th1} and~\ref{second} by proving low degree polynomial (on both $k$ and $|V(H)|$) bounds for more general instantiations of $H$~\cite{2013arXiv1305.7112R}.


\end{document}